\documentclass[aps,prx,a4paper,twocolumn,superscriptaddress,longbibliography,10pt]{revtex4-1}

\usepackage{graphicx,bm,amssymb,physics,mathtools,mathrsfs, amsthm,color,xcolor,bbm,appendix}

\definecolor{magenta(dye)}{rgb}{0.79, 0.08, 0.48}
\definecolor{mediumred-violet}{rgb}{0.73, 0.2, 0.52}
\definecolor{neonfuchsia}{rgb}{1.0, 0.25, 0.39}
\definecolor{orange-red}{rgb}{1.0, 0.27, 0.0}
\definecolor{americanrose}{rgb}{1.0, 0.01, 0.24}
\definecolor{awesome}{rgb}{1.0, 0.13, 0.32}
\definecolor{blue}{rgb}{0.0, 0.0, 1.0}
\definecolor{cadmiumred}{rgb}{0.89, 0.0, 0.13}
\definecolor{candyapplered}{rgb}{1.0, 0.03, 0.0}
\definecolor{electricultramarine}{rgb}{0.25, 0.0, 1.0}
\definecolor{intblue}{rgb}{0.0, 0.18, 0.65}
\definecolor{navyblue}{rgb}{0.0, 0.0, 0.5}
\definecolor{jade}{rgb}{0.0, 0.66, 0.42}

\usepackage[colorlinks=true,linkcolor=jade,citecolor=intblue, urlcolor=cadmiumred, plainpages=false,pdfpagelabels]{hyperref}

\def\Tr{\operatorname{Tr}}

\newcommand{\mc}[1]{\mathcal{#1}}
\newcommand{\mf}[1]{\mathfrak{#1}}
\newcommand{\wt}[1]{\widetilde{#1}}

\newcommand{\qh}[1]{\hat{q_{#1}}}
\newcommand{\ph}[1]{\hat{p_{#1}}}
\newcommand{\io}{\mathfrak{i}}

\newcommand{\tm}{\mathrm{th}}

\newcommand{\diag}[1]{\mathrm{diag}\{#1\}}
\newcommand{\Sp}{\mathrm{Sp}(2n,\mathbb{R})}
\newcommand{\Mr}[1]{\mathrm{#1}}
\newcommand{\mb}[1]{\mathbb{#1}}
\newcommand{\ord}[1]{O\left(n^{#1}\right)}

\newcommand{\Norm}[2]{\left|\left|{#1}\right|\right|_{#2}}
\newcommand{\NormSmall}[2]{||{#1}||_{#2}}
\newcommand{\Mod}[1]{\left|{#1}\right|}

\newtheorem{theorem}{Theorem}

\newtheorem{definition}{Definition}
\newtheorem{lemma}{Lemma}
\newtheorem{proposition}{Proposition}
\newtheorem{remark}{Remark}

\begin{document}

\title{Gaussian work extraction from random Gaussian states is nearly impossible}

\author{Uttam Singh}  
\email{uttam@iiit.ac.in}
\affiliation{Centre for Quantum Information and Communication, \'{E}cole polytechnique de Bruxelles,  CP 165, Universit\'{e} libre de Bruxelles, 1050 Brussels, Belgium}
\affiliation{Center for Theoretical Physics, Polish Academy of Sciences, Aleja Lotnikow 32/46, 02-668 Warsaw, Poland}
\affiliation{Centre of Quantum Science and Technology, International Institute of Information Technology, Hyderabad 500032, India}
\author{Jaros\l{}aw K. Korbicz}  
\email{jkorbicz@cft.edu.pl}
\affiliation{Center for Theoretical Physics, Polish Academy of Sciences, Aleja Lotnikow 32/46, 02-668 Warsaw, Poland}
\author{Nicolas J. Cerf}  
\email{ncerf@ulb.ac.be}
\affiliation{Centre for Quantum Information and Communication, \'{E}cole polytechnique de Bruxelles,  CP 165, Universit\'{e} libre de Bruxelles, 1050 Brussels, Belgium}
\affiliation{James C. Wyant College of Optical Sciences, University of Arizona, Tucson, Arizona 85721, USA}

%\date{\today}

\begin{abstract}
Quantum thermodynamics can be naturally phrased as a theory of quantum state transformation and energy exchange for small-scale quantum systems undergoing thermodynamical processes, thereby making the resource theoretical approach very well suited. A key resource in thermodynamics is the extractable work, forming the backbone of thermal engines. Therefore it is of interest  to characterize quantum states based on their ability to serve as a source of work. From a near-term perspective, quantum optical setups turn out to be ideal test beds for quantum thermodynamics; so it is important to assess work extraction from quantum optical states. Here, we show that Gaussian states are typically useless for Gaussian work extraction. More specifically, by exploiting the ``concentration of measure'' phenomenon, we prove that the probability that the Gaussian extractable work from a zero-mean energy-bounded multimode random Gaussian state is nonzero is exponentially small. This result can be thought of as an $\epsilon$-no-go theorem for work extraction from Gaussian states under Gaussian unitaries, thereby revealing a fundamental limitation on the quantum thermodynamical usefulness of Gaussian components.
\end{abstract}

\maketitle

%%%%%%%%%%%%%%%%

{\it Introduction.---}In the wake of the rapid technological advancements making the control and efficient manipulation of single quantum systems experimentally possible, it has become necessary to address the energetics of nanoscale devices \cite{Scovil1959, Geusic1967, Howard1997, Scully2002, Scully2003, Hanggi2009, Dechant2014, Rossnage2016, Klatzow2019}. Quantum thermodynamics is a burgeoning field of research broadly aimed at systematically addressing this question and, in particular, at challenging the applicability of classical thermodynamics at atomic scales, where quantum effects are inescapable \cite{Thermo2018, Deffner2019, Goold2016, Vinjanampathy2016, Lostaglio2019, Talkner2020}. A number of approaches to a theory of quantum thermodynamics have been developed, including, notably, a quantum resource-theory-based formalism \cite{Michal2013, Brandao2013, Skrzypczyk2014, Brandao2015b, Singh2019, Singh2021}, a purely information theoretic framework \cite{Bera2017, Bera2019}, and open-systems dynamics \cite{Alicki1979, Uzdin2015} (see also a recent book \cite{Thermo2018}). To complement the theoretical efforts towards quantum thermodynamics, there are also exciting new experiments \cite{Rossnage2016, Clos2016, Klatzow2019, Maslennikov2019} that confirm the distinctive features of quantum engines that have been theoretically predicted. Furthermore, quantum effects have been shown to offer advantages in charging quantum batteries \cite{Ferraro2018} and in heat bath algorithmic cooling \cite{Schulman2005, Rodriguez2017}.

Despite tremendous experimental progress in designing quantum thermal machines, quantum thermodynamics is still largely a theoretical endeavor, and more experimental models are needed to confirm the theoretical predictions. It is well established that Gaussian quantum optical states can readily be prepared in the laboratory and Gaussian quantum operations can be implemented efficiently; hence quantum optical setups form a uniquely suited test bed for quantum thermodynamics (see, e.g., Ref. \cite{Dechant2014}). Given that these are central features of quantum thermodynamics, 
work extraction and battery charging have then been investigated in Refs. \cite{Brown2016,Friis2018precisionwork} when restricted to Gaussian operations. More generally, a theory of Gaussian work extraction from multipartite Gaussian states has also been developed in Ref. \cite{Singh2019}. Interestingly, the total amount of work that can be extracted using Gaussian unitaries from a (zero-mean) multipartite Gaussian state was proven to be equal to the difference between the trace and symplectic trace of the covariance matrix \cite{Singh2019} [see Eq. \eqref{eq:loc-gaus-work}]. In order to benchmark the experimental usefulness of such a Gaussian framework for quantum thermodynamics, it is therefore essential to resolve the question of what is the amount of work that can be extracted with Gaussian unitaries if we start from a random multimode Gaussian state? Here, we solve this question by exploiting the ``concentration of measure'' phenomenon~\cite{Ledoux2005}, which states, broadly speaking, that a sufficiently smooth function on a measurable probability space concentrates around its expected value (see also Refs. \cite{Anderson2009, Hayden2006, Popescu2006,  Adesso2006,  Serafini2007c, Serafini2007, Collins2006, Collins2010, Collins2013, Raginsky2013, Fukuda2019b}).

%The concentration of measure phenomenon has been successfully used, e.g., in mathematics \cite{Ledoux2005, Anderson2009} and quantum information theory \cite{Hayden2006, Popescu2006,  Adesso2006,  Serafini2007c, Serafini2007, Collins2006, Collins2010, Collins2013, Raginsky2013, Fukuda2019b} as  a mathematical tool to address a multitude of questions in a generic way.

We start by introducing a procedure to sample energy-bounded random covariance matrices corresponding to a uniform measure on the set of multipartite Gaussian states, following Refs.~\cite{Serafini2007, Fukuda2019b}. As a first technical result, we then show that such random covariance matrices are, typically, locally thermal. Building on this, we prove that the probability that the Gaussian extractable work from a zero-mean energy-bounded random Gaussian state is nonzero is exponentially small. This can be interpreted as an $\epsilon$-no-go theorem for work extraction from random Gaussian states under Gaussian unitaries (see Fig.~\ref{fig:f1}), where $\epsilon>0$ denotes the  work that could potentially be extracted. 
We then discuss the impact of this fundamental near  impossibility on quantum thermodynamics in the Gaussian regime.

\begin{figure}
\centering
\includegraphics[width=\columnwidth]{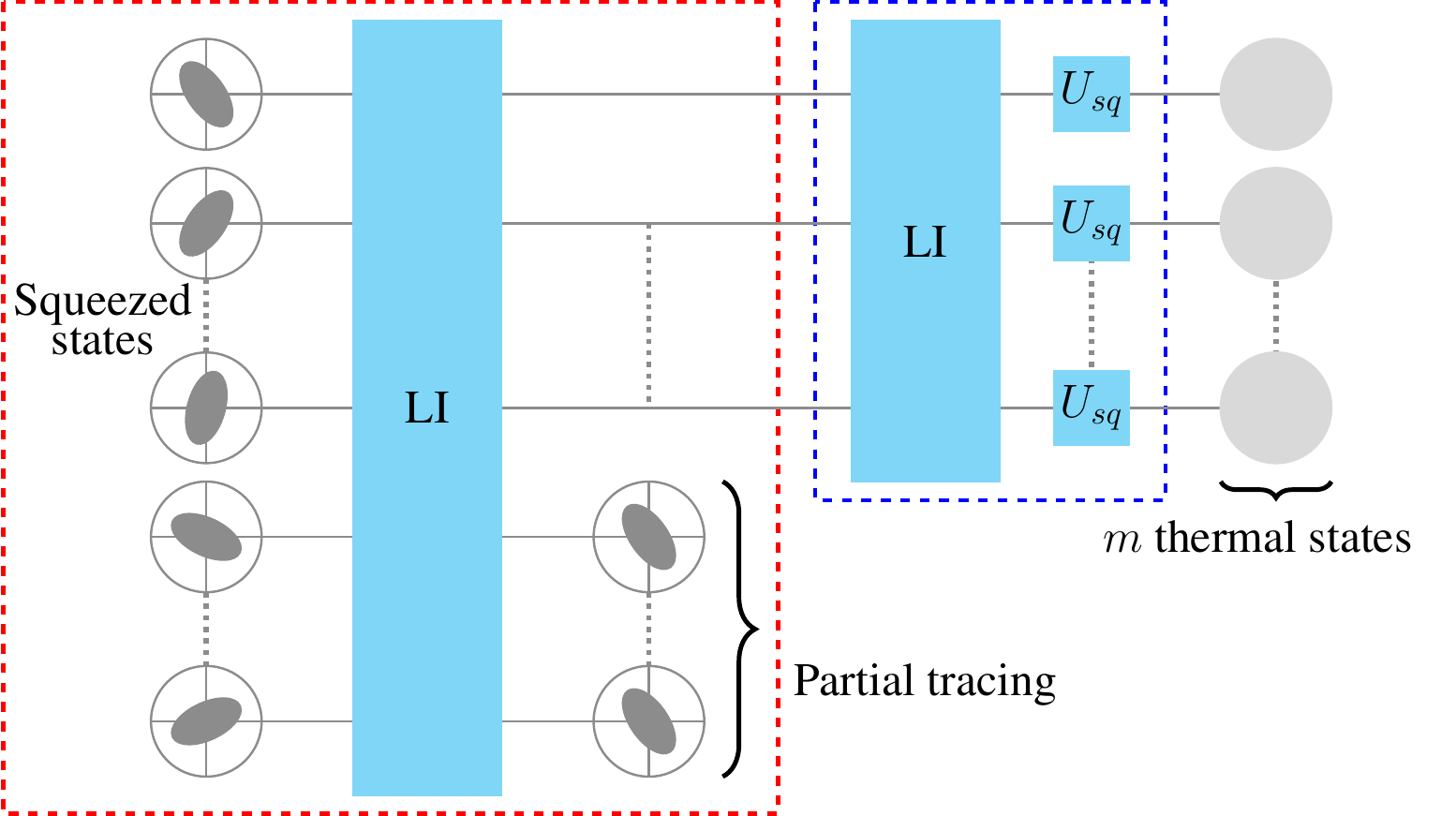}
\caption{Schematic of Gaussian work extraction from random Gaussian states. The dashed red rectangle represents the preparation of energy-bounded $m$-mode random Gaussian mixed states starting from the tensor product of $n$~random squeezed vacuum states, while the dashed blue rectangle represents Gaussian work extraction. In both rectangles, the LI box stands for a linear interferometer (an array of beam splitters and phase shifters). Our main result is the proof that the state emerging from the dashed red rectangle is typically a product of thermal states and hence no work can be extracted by Gaussian means (see Theorem~\ref{th:main-conc}).}
\label{fig:f1}
\end{figure}

%%%%%%%%%%%%%%%%
\medskip
{\it Gaussian states.---}Consider an $n$-mode bosonic system described by Hilbert space $\mc{H}_n:=L^2(\mb{R})^{\otimes n}$ and let ${\bf\hat{x}}=\left(\qh{1},\cdots, \qh{n}, \ph{1},\cdots,\ph{n}\right)$ be the canonical position and momentum operators. They satisfy the canonical commutation relations $
\left[\hat{x}_i,\hat{x}_j\right]=\mathfrak{i}\, \Omega_{ij} \, \mb{I}_{\mc{H}_n} (i,j=1,\cdots, 2n)$, where we set $\hbar=1$ and $\mf{i}=\sqrt{-1}$, 
\begin{align}
\pmb\Omega=\begin{pmatrix}
0 & \mb{I}_n\\
-\mb{I}_n & 0
\end{pmatrix},
\end{align}
and $\mb{I}_n$ is an $n\times n$ identity matrix. The  Hamiltonian of the $i$th mode is given by $\hat{H}_i:=(\hat{q}_i^2+\hat{p}_i^2)/2$ assuming that the angular frequencies of all modes are equal to $1$. For an arbitrary $n$-mode state $\rho$, the $2n$-dimensional mean vector (or coherence vector) ${\bf\bar{x}}$ is defined as
${\bf\bar{x}}:=\left\langle {\bf\hat{x}}\right\rangle =\Mr{Tr}(\rho \, {\bf\hat{x}})$, 
where the angular bracket $\langle\bullet\rangle$ denotes the expectation value with respect to $\rho$. Similarly, the $2n\times 2n$ real positive-definite covariance matrix $\pmb{\Gamma}$ of state $\rho$ is defined via the second-order moments as $
\Gamma_{ij}:=\frac{1}{2}\left\langle \left\{  \hat{x}_{i}-\langle\hat{x}_{i}\rangle,\hat{x}_{j}-\langle\hat{x}_{j}\rangle\right\}  \right\rangle $, 
where $\{ \bullet , \bullet \}$ is the anti commutator.  Gaussian states are states whose characteristic function is Gaussian; hence they are completely described by their mean vector and covariance matrix \cite{Weedbrook2012}. For example, a thermal state is a Gaussian state with ${\bf\bar{x}}=0$ and $\pmb{\Gamma}=(\bar{n}+1/2)\, \mb{I}_{2n}$, where $\bar{n}$ is the average photon number per mode.

%%%%%%%%%%%%%%%%
\medskip
{\it Gaussian unitaries.---}Gaussian unitaries are defined as unitaries that map Gaussian states onto Gaussian states. In particular, a Gaussian unitary $U$ in state space induces an affine map  $(\mc{S}, {\bf d}):{\bf \hat{x}}\rightarrow \mc{S}{\bf\hat{x}}+{\bf d}$ in the space of  quadrature operators ${\bf\hat{x}}$, where  $\mc{S}\in \Sp$ is a $2n\times 2n$ real symplectic matrix (such that $\mc{S}\, \pmb\Omega \, \mc{S}^T=\pmb\Omega$) and ${\bf d}$ is a $2n$-dimensional real vector (displacement vector) \cite{Weedbrook2012}. Thus a Gaussian unitary can be written as $U_{\mc{S}, {\bf d}}= D_{\bf d} \, U_{\mc{S}}$, where $U_{\mc{S}}$ corresponds to the symplectic map  ${\bf \hat{x}}\rightarrow \mc{S}\, {\bf \hat{x}}$ and the Weyl operator $D_{\bf d}$ corresponds to the map ${\bf \hat{x}}\rightarrow {\bf \hat{x}}+{\bf d}$. Under Gaussian unitaries, the first- and second-order moments transform as ${\bf \bar{x}}\rightarrow\mc{S}{\bf \bar{x}}+{\bf d}$ and $\pmb\Gamma\rightarrow\mc{S}\pmb\Gamma\mc{S}^T$. Of special importance to us are the energy-conserving (or passive) Gaussian unitaries, which induce orthogonal symplectic transformations $\mc{S}\in \Sp\cap \Mr{O}(2n) \equiv \Mr{K}_n$ on the quadrature operators, where $\Sp$ is the group of real $2n\times 2n$ symplectic matrices and $\Mr{O}(2n)$ is the group of real orthogonal $2n\times 2n$ matrices.  Physically, passive Gaussian unitaries comprise all linear-optical circuits, also called as linear interferometers (LIs). The other Gaussian unitaries that are relevant here are squeezers. For example, a single-mode squeezer induces the symplectic transformation $\mc{S}=\diag{z,z^{-1}}$, where $z=e^r$ with $r\in\mb{R}$ being the squeezing parameter.

 %%%%%%%%%%%%%%%%
\medskip
{\it Gaussian extractable work.---}The energy of an arbitrary quantum state only depends on the first two moments, ${\bf \bar{x}}$ and $\pmb\Gamma$. In particular, the energy of an $m$-mode state $\rho$ is simply given by $\Mr{Tr}[\rho\, \hat{H}]=\frac{1}{2} \left(\mathrm{Tr}[\pmb\Gamma] +  |\bf\bar{x}|^2\right)$, where $\hat{H}=\sum_{i=1}^m\hat{H}_i$ is the total Hamiltonian \cite{Weedbrook2012}. Now, for any state $\rho$ (Gaussian or otherwise), we define the Gaussian extractable work as the maximum decrease in energy under Gaussian unitaries. Given the decoupled structure of Gaussian unitaries as $U_{\mc{S}, {\bf d}} = D_{\bf d} \, U_{\mc{S}} = U_{\mc{S}} \, D_{\bf d'}$, we can always separate the Gaussian extractable work into two components associated, respectively, with $\mc{S}$ and ${\bf d'}$. Starting from a state with $ {\bf \bar{x}} \ne 0$, it is trivial to extract work first via displacement $D_{\bf d'}$ up to the point where ${\bf \bar{x}}=0$, thereby making this component uninteresting (it can be viewed as classical). Thus we may restrict ourselves to zero-mean states with no loss of generality. Furthermore, using the Bloch-Messiah decomposition \cite{Arvind1995, Braunstein2005}, $U_{\mc{S}}$ can be written as the concatenation of a linear interferometer (LI), a layer of single-mode squeezers, and a second LI. Since the latter leaves the energy unchanged, we may disregard it. Hence the Gaussian extractable work from a zero-mean state $\rho$ is given by (see dashed blue rectangle in Fig.~\ref{fig:f1})
\begin{align}
W(\rho)&= \max_{\Mr{LI},U_{sq}}\mathrm{Tr}\left[\hat{H}\left(\rho-U_{sq} \, \Mr{LI}[  \rho ] \, U_{sq}^{\dagger} \right)\right],
\end{align} 
where $U_{sq}$ denotes the squeezing unitary which corresponds to the tensor product of single-mode squeezers. This maximization yields a particularly simple expression in the phase-space picture, namely  \cite{Singh2019},
\begin{align}
\label{eq:loc-gaus-work}
W(\rho) \equiv \mc{W}(\pmb\Gamma)&=\frac{1}{2}\left(\Mr{Tr}\left[ \pmb\Gamma\right] - \Mr{STr}\left[\pmb\Gamma \right]\right),
\end{align} 
where $\Mr{STr}[\pmb\Gamma]$ denotes the symplectic trace of the covariance matrix $\pmb\Gamma$, i.e.,  the sum of all its symplectic eigenvalues. Note that the Gaussian extractable work $W(\rho)$ is solely a function of $\pmb\Gamma$, so we simply note it $\mc{W}(\pmb\Gamma)$.

%%%%%%%%%%%%%%%%
\medskip
{\it Random sampling of Gaussian states.---}In order to establish the typical behavior of the Gaussian extractable work, we need to give a prescription for the random sampling of covariance matrices.
Consider first an $n$-mode pure Gaussian state. Its covariance matrix $\pmb\Gamma$ can be obtained by applying some Gaussian unitary to the vacuum state; that is, it can be written as $\pmb\Gamma=\mc{S}\mc{S}^T/2$, where $\mc{S}\in\Sp$. From the Bloch-Messiah decomposition, we have $\mc{S}={\bf O} \, [Z({\bf z})\oplus Z^{-1}({\bf z})] \, {\bf O}'$, where ${\bf O}, {\bf O}' \in \Mr{K}_n$ and $Z({\bf z})\oplus Z^{-1}({\bf z})$ is a collection of $n$ single-mode squeezers and $Z({\bf z})=\diag{z_1,\cdots, z_n}$ with $z_i\geq 1$ for all $1\leq i\leq n$. Therefore the covariance matrix is written as
\begin{align}
\pmb\Gamma =\frac{1}{2} {\bf O} \, [J({\bf z})\oplus J^{-1}({\bf z})]  \, {\bf O}^T,
\end{align}
where $J({\bf z})=Z({\bf z})^2$. Now, to define a random covariance matrix, we need to sample ${\bf O}$ and $J({\bf z})$ with appropriate probability measures on their respective spaces. While $\Mr{K}_n$ is a compact space and admits an invariant Haar measure, the space of matrices $J({\bf z})$ is not compact and does not admit a natural invariant  normalizable measure. To properly define the sampling of a random $\pmb\Gamma$, we need some compactness constraint, which can be provided by imposing an energy bound, $\mathrm{Tr}[\pmb\Gamma] \leq 2E$, where $E$ is fixed (remember that ${\bf \bar{x}}=0$). Following Refs. \cite{Serafini2007, Fukuda2019b}, we can sample random matrices $J({\bf z})$ by randomly choosing the vector ${\bf z}$ via the flat Lebesgue measure on the set
\begin{align*}
\mc{G}_{E}:=\left\{(z_1,\cdots, z_n) \,\Big|\,  z_i\geq 1 ~\text{and}~ \sum_{i=1}^n\left(z_i^2+z_i^{-2}\right)\leq 4E \right\};
\end{align*}
%We sample ${\bf z}$ via
that is, we define
the measure $\Mr{d}\mu_{\bf z}=\Mr{d}{z_1}\cdots \Mr{d}{z_n}/{\Mr{Vol}(\mc{G}_E)}$, where $\Mr{Vol}(\mc{G}_E)$ is the volume of $\mc{G}_E$ \footnote{We will not be needing this probability measure in this Research Letter as we will prove our main result for any choice of vector ${\bf z}$ satisfying the energy constraint (see also Ref. \cite{Fukuda2019b}).}.
%%%%%%
% This will automatically be referred as "Note1". We can cite this footnote later as \cite{Note1}
%%%%%%
Then, using the fact that $\Mr{K}_n$ is isomorphic to the complex unitary group $\Mr{U}(n):=\{U\in\mb{C}^{n\times n}:U^\dagger U=\mb{I}_n\}$, we generate random ${\bf O}\in \Mr{K}_n$ via the invariant measure on $\Mr{K}_n$ induced by the Haar measure on $\Mr{U}(n)$. Second, a natural probability measure on the set of energy-bounded random Gaussian mixed states can be induced by performing a partial trace on the random pure state where ${\bf z}$ and ${\bf O}$ are sampled according to the above measures \footnote{We note that an alternative way to define a probability measure on the set of nonzero-mean energy-bounded Gaussian mixed states has been considered in Ref \cite{Lupo2012}.}. 
%%%%%%
% This will automatically be referred as "Note2". We can cite this footnote later as \cite{Note2}
%%%%%%
Note that a Gaussian purification argument  can be used as illustrated in Fig.~\ref{fig:f2} in order to show that the full state can be taken to be pure with no loss of generality (see Proposition \ref{prop:puri} of the Supplemental Material \footnote{See Supplemental Material for additional results, detailed proofs, and extended discussions.}).
%%%%%%
% This will automatically be referred as "Note3". So citing supplementary material is done through \cite{Note3}
%%%%%%
The partial trace of $N=n-m$ out of $n$ modes (see Fig.~\ref{fig:f2}) corresponds to the map $\pmb\Pi_{m,N}$ on the covariance matrix defined as
\begin{align}
\label{eq:mixed-rand}
 \pmb\Gamma\mapsto  \pmb\Gamma_m  := \pmb\Pi_{m,N}\, \pmb\Gamma \, \pmb\Pi_{m,N}, 
\end{align}
where $\pmb\Pi_{m,N}=\wt{\pmb\Pi}\oplus \wt{\pmb\Pi}$ and $\wt{\pmb\Pi}=\diag{\overbrace{1,\cdots,1}^{m},\overbrace{0,\cdots,0}^{N=n-m}}$. Let $\mc{L}_{m,E}$ be the resulting set of covariance matrices $ \pmb\Gamma_m$ for $m$-mode energy-constrained Gaussian mixed states, i.e., $\Mr{Tr}[\pmb\Gamma_m]\leq 2E$ for $\pmb\Gamma_m\in \mc{L}_{m,E}$.  We shall now establish the typicality of the extractable work in $\mc{L}_{m,E}$.

\begin{figure}
\centering
\includegraphics[width=\columnwidth]{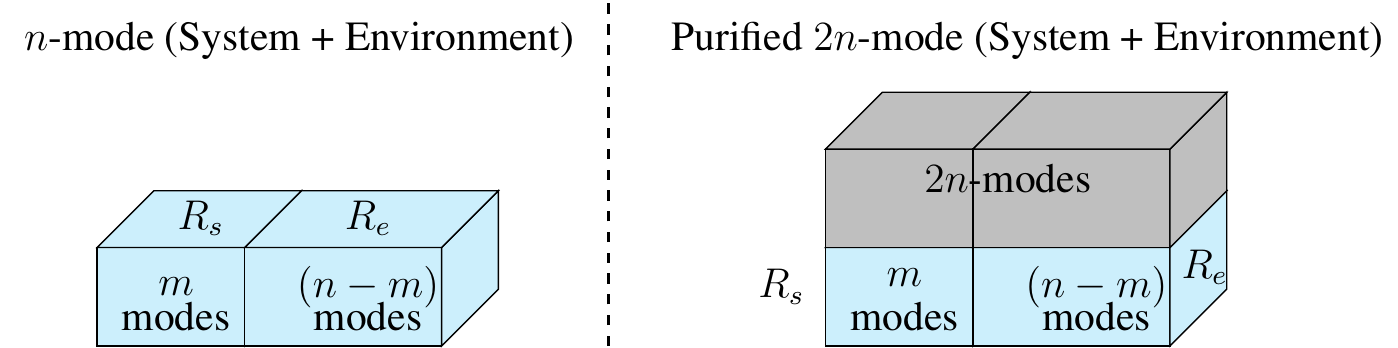}
\caption{The schematic on the left depicts the $m$-mode system of interest (from which work is extracted) denoted as the region $R_s$, which is part of an $n$-mode system. The remaining $N=n-m$ modes constitute the environment, denoted as the region $R_e$. Note that in the most general scenario, the state of the $n$-mode region $R_s+R_e$ can be mixed. However, from Proposition 1 of the Supplemental Material \cite{Note3}, an energy-bounded mixed Gaussian state can be purified into an energy-bounded pure Gaussian state. The schematic on the right represents such a purification of the $n$-mode region into a $2n$-mode region. Now, the state of the $m$-mode system of interest is obtained by performing a partial trace of $(2n-m)$ modes on a pure $2n$-mode state.}
\label{fig:f2}
\end{figure}

%%%%%%%%%%%%%%%%
\medskip
{\it Typicality of Gaussian extractable work.---}We consider a physical scenario where the system of interest is interacting with an inaccessible large environment. This is a common setting in the theory of decoherence  responsible for loss of coherence in quantum systems as a consequence of the partial trace of the environment \cite{Schlosshauer2007}. In particular, in our case the system of interest is a zero-mean energy-bounded $m$-mode Gaussian system embedded in a large Gaussian environment comprising $N \gg m$ modes, so we must characterize the scaling in $n=m+N$ of the energy constraint on the full $n$-mode system as follows.

\begin{definition}
\label{def:bound-rand}
An $m$-mode random Gaussian state is said to be polynomially energy bounded of degree $\beta\geq 0$ if it results from partial  tracing an $n$-mode random Gaussian pure state over $N=n-m$ modes and its covariance matrix can be written as [cf. Eq.~\eqref{eq:mixed-rand}]
\begin{align}
\label{eq:rand-m+n}
\pmb\Gamma_{m}:=\frac{1}{2}\, \pmb\Pi_{m,N} \, {\bf O}_{n} \, \wt{J}_{n}({\bf z}) \, {\bf O}_{n}^T \, \pmb\Pi_{m,N},
\end{align}
where ${\bf O}_n\in \Mr{K}_n$ and the sequence $\wt{J}_{n}({\bf z}):=J_{n}({\bf z})\oplus J^{-1}_{n}({\bf z})$ is such that 
$\NormSmall{\wt{J}_{n}({\bf z})}{\infty}=O\left(n^{\beta}\right)$  \footnote{We say that a function $f:\mb{R}\rightarrow\mb{R}^+$ scales as $O(n)$ when for all sufficiently large $x\in\mb{R}$, there exists a constant $c>0$ such that $f(x)\leq cn$. Similarly, we say that a function $f:\mb{R}\rightarrow\mb{R}^+$ scales as $\Omega(n)$ when for all sufficiently large $x\in\mb{R}$, there exists a constant $c>0$ such that $f(x)\geq cn$.}.
%%%%%%
% This will automatically be referred as "Note4". We can cite this footnote later as \cite{Note4}
%%%%%%
\end{definition}
Since ${\bf O}_n$ is passive, the $n$-mode random Gaussian pure state of covariance matrix $\pmb\Gamma_{n}:=\frac{1}{2}{\bf O}_n\wt{J}_{n}({\bf z}) {\bf O}_n^T$ has energy $E_{\beta}=\ord{\beta+1}$. Then, it can be used to generate a polynomially energy-bounded $m$-mode random Gaussian mixed state according to Definition \ref{def:bound-rand} (see Fig. \ref{fig:f2}). 

\medskip
Before proving our main result, we first need to establish the asymptotic behavior of the eigenspectrum and symplectic eigenspectrum of the energy-bounded random covariance matrices $\pmb\Gamma_m$ from Definition \ref{def:bound-rand}, which is the content of the following two lemmas (their proofs are provided in the Supplemental Material \cite{Note3}).

\begin{lemma}
\label{lem:reg-spectra}
Let $\pmb\Gamma_{m}$ be the covariance matrix of an $m$-mode polynomially energy-bounded random Gaussian state of degree $\beta<1/4$ (see Definition \ref{def:bound-rand}). For universal constants $\gamma,\wt{\gamma}>0$ such that $\epsilon > 2\gamma n^{\beta-1}$, the eigenvalues $\{\lambda_i\}_{i=1}^{2m}$ of $\pmb\Gamma_m$ converge in probability to $\nu_{\tm}$, i.e.,
\begin{align}
\Mr{Pr}\left[ \sum_{i=1}^{2m}\left(\lambda_i-\nu_{\tm}\right)^2>\epsilon \right]\leq \exp\left[ -\wt{\gamma} \epsilon^2 n^{1-4\beta} \right],
\end{align}
%where $\nu_\tm$ is the unique eigenvalue of the thermal covariance matrix with the same average energy as $\pmb\Gamma_{m}$.
where 
%$\nu_\tm$ is the average energy per mode of the $2n$-mode input pure state, i.e., 
$\nu_\tm=\Tr[\wt{J}_{2n}({\bf z})]/(8n)$.
\end{lemma}

\begin{lemma}[\cite{Fukuda2019b}]
\label{lem:symp-sp}
Let $\pmb\Gamma_{m}$ be the covariance matrix of an $m$-mode polynomially energy-bounded random Gaussian state of degree  $\beta<1/8$ (see Definition \ref{def:bound-rand}). For universal constants $C,c>0$ such that $\epsilon> Cn^{\beta-1}$, the symplectic eigenvalues $\{\nu_i\}_{i=1}^m$ of $\pmb\Gamma_m$ converge in probability to $\nu_{\tm}$, i.e.,
\begin{align}
\Mr{Pr}\left[ \sum_{i=1}^m\left(\nu_i^2-\nu_{\tm}^2\right)^2>\epsilon \right]\leq \exp\left[ -c\epsilon^2  n^{1-8\beta} \right],
\end{align}
%where $\nu_\tm$ is the unique eigenvalue of the thermal covariance matrix with the same average energy as $\pmb\Gamma_{m}$.
where 
%$\nu_\tm$ is the average energy per mode of the $2n$-mode input pure state, i.e., 
$\nu_\tm=\Tr[\wt{J}_{2n}({\bf z})]/(8n)$.
\end{lemma}

Our main result lies in the following $\epsilon$-no-go theorem for the Gaussian extractable work from a polynomially energy-bounded random Gaussian state.

\begin{theorem}[typicality of the Gaussian extractable work]
\label{th:main-conc}
Let $\pmb\Gamma_{m}$ be the covariance matrix of an $m$-mode polynomially energy-bounded random Gaussian state of degree  $\beta<1/8$ resulting from performing a partial trace on a random $n$-mode state as in Definition \ref{def:bound-rand}. Then, for universal constants $\tilde{c}, C>0$ such that $\epsilon > \sqrt{2Cmn^{\beta-1}}$, the Gaussian extractable work $\mc{W}\left(\pmb\Gamma_m\right)$ satisfies
\begin{align}
\label{eq:thmain}
\Mr{Pr}\left[\mc{W}\left(\pmb\Gamma_m\right)>\epsilon \right]\leq \exp\left[-\tilde{c} \epsilon^4 n^{1-8\beta} \right].
\end{align}
\end{theorem}

As noted above, the energy of $\pmb\Gamma_{n}$ scales as $\ord{\beta +1}$, so the energy of $\pmb\Gamma_{m}$ scales as $\ord{\beta}$ after tracing over the $N=\Omega(n)$ environmental modes \cite{Note2}. The extractable work tolerance $\epsilon$ in Theorem \ref{th:main-conc}  thus scales as $\sqrt{\frac{m}{n^{1-\beta}}}$, which goes to zero in the limit $n\rightarrow\infty$ as $\beta<1/8$. Furthermore, the right-hand side of Eq. \eqref{eq:thmain} in Theorem \ref{th:main-conc} exponentially goes to zero in the limit $n\rightarrow\infty$ as  $\beta<1/8$. These two elements make Theorem \ref{th:main-conc} meaningful.

%%%%%%%%%%%%%%%%
\medskip
{\it Proof of the theorem.---}Using the Gaussian purification argument (see Proposition 1 of the Supplemental Material \cite{Note3}), i.e., $n\mapsto 2n$, Eq. \eqref{eq:rand-m+n} becomes
\begin{align}
\label{eq:new-cov}
\pmb\Gamma_{m}:=\frac{1}{2} \pmb\Pi_{m,N+n}{\bf O}_{2n} \wt{J}_{2n}({\bf z}) {\bf O}_{2n}^T \pmb\Pi_{m,N+n},
\end{align}
where ${\bf O}_{2n}\in \Mr{K}_{2n}$ is a $4n\times 4n$ orthogonal symplectic matrix. Also, we have (see Refs. \cite{Note3, Fukuda2019b})
\begin{align}
\label{eq:new-cov-next}
{\bf O}_{2n}= P\begin{pmatrix}
U & 0\\
0 & U^*
\end{pmatrix}P^{-1},
\end{align}
where $
P=\frac{1}{\sqrt{2}}\begin{pmatrix}
\mb{I}_{2n} & \io \mb{I}_{2n}\\
\io \mb{I}_{2n} & \mb{I}_{2n}
\end{pmatrix}$ and $U$ is a $2n\times 2n$ random unitary matrix. We define the function $T_m:\Mr{U}(2n)\rightarrow\mb{R}$ of random unitary matrices as
\begin{align}
T_m(U):=\sum_{k=1}^{2m}\left(\lambda_k-\nu_{\tm}\right)^2  ,
\end{align}
where $\{\lambda_k\}_{k=1}^{2m}$ are the eigenvalues of $\pmb\Gamma_m\equiv \pmb\Gamma_m(U)$ as defined from Eqs.~\eqref{eq:new-cov} and \eqref{eq:new-cov-next} and $\nu_\tm$ is the average energy per mode of the $2n$-mode input pure state with covariance matrix $\frac{1}{2} {\bf O}_{2n} \wt{J}_{2n}({\bf z}) {\bf O}_{2n}^T $, that is, $\nu_\tm=\Tr[\wt{J}_{2n}({\bf z})]/(8n)$. By exploiting the concentration of measure, Lemma \ref{lem:reg-spectra} then gives us an exponentially small upper bound on $\Mr{Pr}\left[ T_m(U)  >\epsilon \right]$. Similarly, we define the function $\mf{T}_m:\Mr{U}(2n)\rightarrow\mb{R}$ as
\begin{align}
\mf{T}_m(U) := 2\sum_{k=1}^{m}\left(\nu_k^2 - \nu_{\tm}^2\right)^2 ,
\end{align}
where $\{\nu_k\}_{k=1}^m$ are the symplectic eigenvalues of $\pmb\Gamma_m (U)$, and use Lemma \ref{lem:symp-sp} to obtain an exponentially small upper bound on 
$\Mr{Pr}\left[ \mf{T}_m(U)  >\epsilon \right]$ (see also Ref. \cite{Fukuda2019b}). 
Thus, together, Lemmas \ref{lem:reg-spectra} and \ref{lem:symp-sp} imply that energy-bounded random Gaussian states are typically locally thermal with the same average energy in each mode (since both the symplectic and regular eigenspectra concentrate around a thermal spectrum).

This is the key to our proof of the near impossibility of Gaussian work extraction from energy-bounded random Gaussian states. Consider the function $\Delta_m:\Mr{U}(2n)\rightarrow\mb{R}$ as $\Delta_m(U):=T_m(U)+\mf{T}_m(U)$.  From Lemmas \ref{lem:reg-spectra} and \ref{lem:symp-sp} we know that both $T(U)$ and $\mf{T}(U)$ are Lipschitz continuous functions on $\Mr{U}(2n)$ \cite{Note3}. Then, for any two unitaries $U,V\in\Mr{U}(2n)$, we have
\begin{align*}
\Mod{\Delta(U)-\Delta(V)}&\leq \Mod{T(U)-T(V)}+\Mod{\mf{T}(U)-\mf{T}(V)}\\
&\leq \ord{4\beta}\Norm{U-V}{2}.
\end{align*}
Thus, $\Delta_m(U)$ is a Lipschitz continuous function on $\Mr{U}(2n)$ with a Lipschitz constant given by $\theta \, n^{4\beta}$, where $\theta$ is a universal constant. Furthermore,  we have
\begin{align*}
\mb{E}_U\Delta(U)&= \mb{E}_UT(U)+\mb{E}_U\mf{T}(U)\\
&=\ord{\beta-1}\leq C n^{\beta-1},
\end{align*}
where $C$ is a universal constant \cite{Note3}. Next, we can exploit the concentration of measure for $\Delta_m(U)$ in a manner similar to that followed for $T_m(U)$ and $\mf{T}_m(U)$. For universal constants $C,c'>0$ such that $\delta>2Cn^{\beta-1}$, we have
\begin{align}
\label{eq:typ-delta}
\Mr{Pr}\left[\Delta(U)\geq \delta\right]&\leq \Mr{Pr}\left[\Delta(U)> \frac{\delta}{2}+\mb{E}_{U}\Delta(U)\right]\nonumber\\
&\leq \exp\left[-\frac{\delta^2 n}{48\theta^2 n^{8\beta}}\right]\nonumber\\
&\leq \exp\left[-c'\delta^2 n^{1-8\beta}\right],
\end{align}
where $c'$ is a universal constant. Then, we express the Gaussian extractable work $\mc{W}\left(\pmb\Gamma_m\right)$ as a function of $\Delta_m(U)$. Noting that $\mc{W}\left(\pmb\Gamma_m\right)= \left|\mc{W}\left(\pmb\Gamma_m\right)\right|$, we have 
\begin{align*}
\mc{W}\left(\pmb\Gamma_m\right)
&=\left|\frac{1}{2}\sum_{k=1}^{2m}(\lambda_k-\nu_{\tm})+\sum_{k=1}^{m}(\nu_{\tm}-\nu_k)\right|\\
&\leq \frac{1}{2}\left(\sum_{k=1}^{2m}\left|\lambda_k-\nu_{\tm}\right|+2\sum_{k=1}^{m}\left|\nu_k^2-\nu_{\tm}^2\right|\right)\\
&\leq \sqrt{m}\sqrt{\sum_{k=1}^{2m}\left|\lambda_k-\nu_{\tm}\right|^2+2\sum_{k=1}^{m}\left|\nu_k^2-\nu_{\tm}^2\right|^2}\\
&=  \sqrt{m \, \Delta_m(U)},
\end{align*}
where the first inequality follows from the triangle inequality and the fact that $\nu_{\tm}+\nu_k\geq 1$. The second inequality follows from the inequality $(\sum_{i=1}^n x_i)^2 \leq n\sum_{i=1}^n x_i^2$. Letting $\delta>2Cn^{\beta-1}$ and using Eq. \eqref{eq:typ-delta}, we have
\begin{align*}
\Mr{Pr}\left[\mc{W}\left(\pmb\Gamma_m\right)\leq  \sqrt{m\delta}\right]
&\geq  \Mr{Pr}\left[\Delta_m(U) \leq \delta \right]\\
&\geq 1-  \exp\left[-c'\delta^2 n^{1-8\beta}\right].
\end{align*}
Furthermore, letting $\epsilon=\sqrt{m\delta}>\sqrt{2Cm n^{\beta-1}}$, we have
\begin{align*}
\Mr{Pr}\left[\mc{W}\left(\pmb\Gamma_m\right)>\epsilon\right]
&\leq \exp\left[-\wt{c}\epsilon^4 n^{1-8\beta}\right],
\end{align*}
where $\wt{c}$ is a universal constant, which concludes our proof.

\medskip
{\it Conclusion.---}We have established the near impossibility---in a strong sense---of extracting work from (zero-mean) polynomially energy-bounded random multimode Gaussian states using Gaussian unitaries. Qualitatively, this follows from the fact that these states are typically locally thermal as a consequence of the ``concentration of measure" phenomenon, so one cannot extract work from such states.
This is a probabilistic statement in the sense that there remains an exponentially small probability that it does not hold. In this regard, our $\epsilon$-no-go theorem slightly contrasts with the well-known Gaussian no-go theorems, e.g., for Gaussian universal quantum computation \cite{Lloyd1999}, for the distillation of entanglement from Gaussian states using Gaussian local  operations and classical communication \cite{Eisert2002,Fiurasek2002,Giedke2002}, or for Gaussian quantum error correction \cite{Niset2009}.

Our findings reveal a fundamental limitation on the processing of Gaussian states using Gaussian operations and show that harnessing quantum thermodynamical processes is typically impossible in the Gaussian regime. This limitation even goes beyond the extractable work as a similar $\epsilon$-no-go theorem can be proven for the single-mode relative entropy of activity (an alternative measure of the distance from a Gaussian thermal state defined in Ref. \cite{Singh2019}) of random Gaussian states; see Ref. \cite{us3}.

Keeping in mind that quantum optical setups are among the leading platforms for experimental quantum thermodynamics, our results point to an essential requirement to consider non-Gaussian (or even perhaps nontypical Gaussian) components in quantum thermodynamics. A natural question that arises in this regard is the following: Which nontypical or non-Gaussian states can get around the $\epsilon$-no-go theorems presented here and hence be useful for work extraction? It would be very interesting to address this question, which, in turn, could open up exciting developments in 
quantum thermodynamics with non-Gaussian resources, a topic that has hardly been explored to date.

On a final note, our results are reminiscent to the de~Finetti theorem for quantum states that are invariant under orthogonal symplectic transformations \cite{Leverrier2009}. Such a de~Finetti theorem states that if we perform a partial trace on a state that obeys this invariance, the resulting state approaches a mixture of products of (independent and identically distributed) thermal states. However, we note that the random states considered here are not, in general, invariant under orthogonal symplectic transformations; yet we are able to show that they concentrate around a product of thermal states with the same mean photon number. In fact, we have a stronger bound here as the error probability is exponentially small, while it is polynomially small in the de~Finetti theorem of Ref.~\cite{Leverrier2009}. On the other hand, our result is only concerned with Gaussian states, while the de~Finetti theorem of Ref.~\cite{Leverrier2009} holds for any state with the right invariance. This suggests that it would be interesting to explore the relationship between de~Finetti theorems and the ``concentration of measure" phenomenon.

\medskip
\begin{acknowledgments}
U.S. and J.K.K. acknowledge support from the Polish National Science Center (NCN; Grant No. 2019/35/B/ST2/01896). N.J.C. acknowledges support from Fonds de la Recherche Scientifique – FNRS under Grant No. T.0224.18 and from the European Union under project ShoQC within ERA-NET Cofund in Quantum Technologies (QuantERA) program.
\end{acknowledgments}

\medskip
\bibliography{generic-work}

% \tcr{+++ Here, we describe the physical scenario where we are interested in establishing the typicality of the extractable work.  }

% \tcr{+++ In the following,  we consider polynomially energy bounded $n$-mode quantum systems with $n=m+N$,  where $m$-modes constitute the system of interest and $N\gg m$ is the number of modes of the environment surrounding the $m$-mode system.}

% \tcr{+++ $\epsilon$-no-go theorem for work extraction from energy bounded zero displacement random Gaussian states under Gaussian unitaries}

%%%%%%%%%%%%%%%%%%%%%%%%%%%%%%%%%%%%%%%%%%%%%%%%%%%%%%%

\clearpage
\newpage
\begin{widetext}

\noindent  {\bf Supplemental Material: Gaussian work extraction from random Gaussian states is nearly impossible}

\medskip
\noindent
%\appendix
In Section \ref{sec:GauPur}, we use a Gaussian purification argument to prove that the total state can be taken pure with no loss of generality in the argument leading to the definition of the set $\mc{L}_{m,E}$.
In Sections \ref{sec:App-Lemma1} and \ref{sec:App-Lemma2}, we state and prove Lemmas \ref{lem:reg-spectra} and \ref{lem:symp-sp}, respectively, which are used in order to prove Theorem \ref{th:main-conc} in the main text.
%   In Section \ref{append:act-one1}, we prove the $\epsilon$-no-go theorem for the single-mode Gaussian activity for random Gaussian states. 
   In Appendix \ref{sec:lip}, we detail the mathematical tools used to prove the results presented here. 
Appendix \ref{sec:Wein} elaborates on the method to compute averages over Haar distributed unitaries following Weingarten calculus.
% Finally, in Appendix \ref{append-avg}, we complete the proof steps of the results that are left unproven in the main text of the Supplemental Material.

%%%%%%%%%%%%%%%

\section{Gaussian purification}
\label{sec:GauPur}

We have the following proposition on the structure of the set $\mc{L}_{m,E}$.

\begin{proposition}[Gaussian purification]
\label{prop:puri}
For any covariance matrix $\pmb\Gamma_m\in \mc{L}_{m,E}$, there exists a $2m$-mode covariance matrix $\pmb\Gamma$ corresponding to pure Gaussian state with $\Mr{Tr}[\pmb\Gamma]\leq 4E$ such that $\pmb\Gamma_m=\pmb\Pi_{m,m}\pmb\Gamma\pmb\Pi_{m,m}$, where $\pmb\Pi_{m,m}=\diag{\overbrace{1,\cdots,1}^{m},\overbrace{0,\cdots,0}^{m}}\oplus \diag{\overbrace{1,\cdots,1}^{m},\overbrace{0,\cdots,0}^{m}}$.
\end{proposition}

\begin{proof}
Given an $m$-mode  covariance matrix $\pmb\Gamma_m$ of system $A$, using Williamson's theorem, we can write
\begin{align}
\label{eq:arb-cov}
\pmb\Gamma_m=\mc{S}\pmb\Gamma_{0m} \mc{S}^T,
\end{align}
where $\mc{S}$ is a symplectic matrix and $\pmb\Gamma_{0m}=\diag{\nu_1,\cdots, \nu_m}\oplus \diag{\nu_1,\cdots, \nu_m}$ with $\nu_i$ being the symplectic eigenvalues. $\pmb\Gamma_{0m}$ is a collection of $m$ single mode thermal states and the mean photon number of the $i$th thermal state is given by $(2\nu_i-1)/2$. It is known that a single mode thermal state can be purified using two mode squeezed state. In particular, a two mode squeezer on systems $A$ and $R$ is described as a symplectic transformation $\mc{S}_{TMS}$ given by
\begin{align}
\mc{S}_{TMS}=\begin{pmatrix}
\cosh r_i \mb{I} & \sinh r_i \sigma_z \\ \sinh r_i \sigma_z & \cosh r_i\mb{I}
\end{pmatrix},
\end{align}
which acts linearly on the two mode quadrature operators ${\bf x}_i= (q_i^A,p_i^A,q_i^R, p_i^R)$. The covariance matrix for two mode vacuum state is given by $\frac{1}{2}(\mb{I}\oplus \mb{I})$, where $\mb{I}$ is a $2\times 2$ identity matrix. Then the two mode squeezed vacuum state is given by
\begin{align}
\pmb\Gamma_{TMS}&=\frac{1}{2}\mc{S}_{TMS}\mc{S}_{TMS}^T\nonumber\\
&=\frac{1}{2}\begin{pmatrix}
\cosh r_i \mb{I} & \sinh r_i \sigma_z \\ \sinh r_i \sigma_z & \cosh r_i \mb{I}
\end{pmatrix}\begin{pmatrix}
\cosh r_i \mb{I} & \sinh r_i \sigma_z \\ \sinh r_i \sigma_z & \cosh r_i \mb{I}
\end{pmatrix}\nonumber\\
&=\frac{1}{2}\begin{pmatrix}
\cosh 2r_i \mb{I} & \sinh 2r_i \sigma_z \\ \sinh 2r_i \sigma_z & \cosh 2r_i\mb{I}
\end{pmatrix}\nonumber\\
&=\begin{pmatrix}
\nu_i \mb{I} & \sqrt{\nu_i^2-1/4} \sigma_z \\ \sqrt{\nu_i^2-1/4} \sigma_z & \nu_i\mb{I}
\end{pmatrix},
\end{align}
where $\cosh 2r_i = 2\nu_i$. We see that indeed removing one mode (the $R$ mode) gives us a thermal state with mean photon number $(2\nu_i-1)/2$. Similarly, we can write purification for $m$ modes. In particular, $\wt{\pmb{\Gamma}}_{2m}$ is a purification of $\pmb\Gamma_{0m}$, where
\begin{align}
\wt{\pmb{\Gamma}}_{2m} &=\begin{pmatrix}
\pmb{\Gamma}_{0m} & V \\ V & \pmb{\Gamma}_{0m}
\end{pmatrix},
\end{align}
 $V=\oplus_{i=1}^n \sqrt{\nu_i^2-1/4}\sigma_z$, and the order of quadrature operators is given by $(q_1^A, p_1^A\cdots, q_m^A, p_m^A,  q_1^R, p_1^R\cdots, q_m^R,p_m^R)$. Thus, we have
 \begin{align}
\pmb{\Gamma}_{2m}'&:=(\mc{S}\oplus\mb{I}_m)\wt{\pmb{\Gamma}}_{2m} (\mc{S}\oplus\mb{I}_m)^T\nonumber\\
 &=\begin{pmatrix}
\mc{S}\pmb{\Gamma}_{0m} \mc{S}^T & \mc{S}V \\ V\mc{S}^T & \pmb{\Gamma}_{0m}
\end{pmatrix}\nonumber\\
&=\begin{pmatrix}
\pmb{\Gamma}_m & \mc{S}V \\ V\mc{S}^T & \pmb{\Gamma}_{0m}
\end{pmatrix},
\end{align}
where $\mc{S}$ is the same as in Eq. \eqref{eq:arb-cov}.  For consistency of the notation, we further need to perform a permutation $P_{\pi}$ on $2m$ indices such that 
\begin{align}
&P_{\pi}(q_1^A, p_1^A\cdots, q_m^A, p_m^A, q_1^R, p_1^R\cdots, q_m^R,p_m^R)^T=(q_1^A,\cdots, q_m^A, q_1^R\cdots, q_m^R, p_1^A,\cdots, p_m^A, p_1^R, \cdots, p_m^R)^T.
\end{align}
Thus, the desired purification is given by $\pmb{\Gamma}_{2m}=P_{\pi} \pmb{\Gamma}_{2m}'P_{\pi}^{-1}$ as $\pmb{\Gamma}_{m}=\pmb\Pi_{m,m}\pmb{\Gamma}_{2m}\pmb\Pi_{m,m}$. Now let $\pmb\Gamma_m\in \mc{L}_{m,E}$, i.e., $\Mr{Tr}[\pmb\Gamma_m]\leq 2E$. The energy corresponding to covariance matrix $\pmb{\Gamma}_{2m}$ is given by
\begin{align}
\frac{1}{2}\Mr{Tr}[\pmb\Gamma_{2m}]&=\frac{1}{2}\Mr{Tr}[\pmb\Gamma_{2m}']\nonumber\\
&=\frac{1}{2}\Mr{Tr}[\pmb\Gamma_{m}]+\frac{1}{2}\Mr{Tr}[\pmb\Gamma_{0m}]\nonumber\\
&\leq \Mr{Tr}[\pmb\Gamma_{m}]\nonumber\\
&\leq 2E,
\end{align}
where we used the fact that $\Mr{Tr}[\pmb\Gamma_{m}]=\Mr{Tr}[\mc{S}\pmb\Gamma_{0m}\mc{S}^T] \geq \min_{\mc{S}\in\Sp}\Mr{Tr}[\mc{S} \pmb\Gamma_{0m}\mc{S}^T]= \Mr{Tr}[\pmb\Gamma_{0m}]$ \cite{Singh2019}. This completes the proof of the proposition.
\end{proof}

%\tcr{ Note that $\pmb\Pi_{m,N}\pmb\Pi_{n,n}= \pmb\Pi_{m,N+n}:=\pmb\Pi\oplus\pmb\Pi$, where $\pmb\Pi=\diag{\overbrace{1,\cdots,1}^{m},\overbrace{0,\cdots,0}^{N+n}} $. WHERE SHOULD WE PLACE THIS SENTENCE? }

%%%%%%%%%%%%%%%

\section{Proof of Lemma \ref{lem:reg-spectra}}
\label{sec:App-Lemma1}

Here, we give a proof of Lemma \ref{lem:reg-spectra}, which establishes the typicality of the eigenspectra as needed to prove our main result on the Gaussian extractable work, Theorem \ref{th:main-conc} of the main text.
%, as well as its counterpart for the single-mode Gaussian activity, Theorem \ref{th:main-act-conc} in Section \ref{append:act-one1}.
First, let us recall the physical procedure to sample a $n$-mode energy-constrained random Gaussian pure state. We sample ${\bf z}$ satisfying the energy constraint, yielding a squeezed vacuum state $\ket{\Psi({\bf z})}=\ket{\psi_{z_1}}\otimes\cdots\otimes\ket{\psi_{z_n}}$ with covariance matrix $J({\bf z})\oplus J^{-1}({\bf z})$, and then we  apply a random ${\bf O}\in\Mr{K}_n$ sampled from the invariant measure on $\Mr{K}_n$. This is done via the isomorphism $F:\Mr{U}(n)\mapsto \Mr{K}_n$ defined as
\begin{align}
\label{eq:o-def}
{\bf O}\equiv F(U):= P\begin{pmatrix}
U & 0\\
0 & U^*
\end{pmatrix}P^{-1},
\end{align}
where $
P=\frac{1}{\sqrt{2}}\begin{pmatrix}
\mb{I}_n & \io \mb{I}_n\\
\io \mb{I}_n & \mb{I}_n
\end{pmatrix}$ and $P^{\dagger}=P^{-1}$. Using the Gaussian purification argument, we actually replace $n$ with $2n$ in the above. Then, the zero-mean random $m$-mode pure

\begin{align}
\label{eq:new-cov6}
\pmb\Gamma_{m}:=\frac{1}{2} \pmb\Pi_{m,N+n}{\bf O}_{2n} \wt{J}_{2n}({\bf z}) {\bf O}_{2n}^T \pmb\Pi_{m,N+n},
\end{align}

\noindent
{\it Restatement of Lemma \ref{lem:reg-spectra}:--}
Let $\pmb\Gamma_{m}$ be the covariance matrix as in Eq.  \eqref{eq:new-cov6}.  Further, assume that $4\beta<1$. For universal constants $\gamma,\wt{\gamma}>0$ such that $\epsilon > 2\gamma n^{\beta-1}$, the eigenvalues $\{\lambda_i\}_{i=1}^{2m}$ of $\pmb\Gamma_m$ converge in probability to $\nu_{\tm}$, i.e.,
\begin{align}
\Mr{Pr}\left[ \sum_{i=1}^{2m}\left(\lambda_i-\nu_{\tm}\right)^2>\epsilon \right]\leq \exp\left[ -\wt{\gamma} \epsilon^2 n^{1-4\beta} \right],
\end{align}
where $\nu_\tm$ is the average energy per mode of the $2n$-mode input pure state, i.e., $\nu_\tm=\Tr[\wt{J}_{2n}({\bf z})]/(8n)$.

\begin{proof}
Let us consider a function $T_m:\Mr{U}(2n)\rightarrow\mb{R}$ of random unitary matrices defined as
\begin{align}
T_m(U)&:=\Mr{Tr}\left[\left(\pmb\Gamma_m-\nu_{\tm}\mb{I}_{2m}\right)^2\right],
\end{align}
where $\pmb\Gamma_m\equiv \pmb\Gamma_m(U)$ is defined by Eq. \eqref{eq:new-cov6} and $\nu_\tm$ is the average energy of the $2n$-mode input pure state with covariance matrix $\frac{1}{2} {\bf O}_{2n} \wt{J}_{2n}({\bf z}) {\bf O}_{2n}^T $. 
Thus we have $\nu_\tm=\Tr[\wt{J}_{2n}({\bf z})]/(8n)$. Also, note that $\nu_\tm$ is uniformly bounded in $n$ by definition of $\wt{J}_{2n}({\bf z}) = J_{2n}({\bf z})\oplus J^{-1}_{2n}({\bf z})$ with $\Norm{\wt{J}_{2n}({\bf z})}{\infty}=O\left(n^{\beta}\right)$.  Since $\pmb\Gamma_m$ is a real symmetric matrix, it can be diagonalized by an orthogonal matrix and we have
\begin{align}
T_m(U)=\sum_{k=1}^{2m}\left(\lambda_k-\nu_{\tm}\right)^2.
\end{align}

Thus, the lemma provides us with an upper bound to the probability $\Mr{Pr}\left[ T_m(U) >\epsilon \right]$. The proof follows from the concentration of measure phenomenon as we show below. First, we compute the average value of the function $T_m(U)$. By definition of $\pmb\Gamma_m$, we have
\begin{align}
2\pmb\Gamma_m=\frac{1}{2}\begin{pmatrix}
\pmb\Pi & \io\pmb\Pi\\
\io\pmb\Pi & \pmb\Pi
\end{pmatrix}\mc{A}(U)
\begin{pmatrix}
\pmb\Pi & \io\pmb\Pi\\
\io\pmb\Pi & \pmb\Pi
\end{pmatrix},
\end{align}
where
\begin{align}
\mc{A}(U)=
\begin{pmatrix}
U A U^T & -\io UBU^\dagger\\
-\io U^*B U^T & -U^* A U^{\dagger}
\end{pmatrix}
\end{align}
with $A=(J_{2n}({\bf z})-J_{2n}^{-1}({\bf z}))/2$ and $B=(J_{2n}({\bf z})+J_{2n}^{-1}({\bf z}))/2$. From Remark \ref{rem:wein} in appendix \ref{sec:Wein}, we have $\mb{E}_{\Mr{U}} U A U^T=0$. Further, it is easy to see that
\begin{align}
&\mb{E}_{\Mr{U}} UBU^\dagger=\Mr{Tr}[B]\frac{\mb{I}}{2n}.
\end{align}
Now, using  $\Norm{B}{1} =\Mr{Tr}[B]=4n\nu_{\tm}$, we have $
\mb{E}_{\Mr{U}}\pmb\Gamma_m=\nu_{\tm}\begin{pmatrix}
\pmb\Pi & 0\\
0 & \pmb\Pi
\end{pmatrix}$. Thus,
\begin{align}
\label{eq:av1}
\Mr{Tr}\left[\mb{E}_{\Mr{U}}\pmb\Gamma_m\right]=2m\nu_{\tm}.
\end{align}
 Moreover,
\begin{align*}
4\pmb\Gamma^2_m=\frac{\io}{2}\begin{pmatrix}
\pmb\Pi & \io\pmb\Pi\\
\io\pmb\Pi & \pmb\Pi
\end{pmatrix}
\begin{pmatrix}
\mc{Q}(U) & -\mc{B}(U)\\
-\mc{B}^*(U) & \mc{Q}^*(U)
\end{pmatrix}
\begin{pmatrix}
\pmb\Pi & \io\pmb\Pi\\
\io\pmb\Pi & \pmb\Pi
\end{pmatrix},
\end{align*}
where
\begin{align*}
&\mc{Q}(U)=-\io U BU^\dagger\pmb\Pi U  A U^T-\io U A U^T\pmb\Pi U^* BU^T;\\
&\mc{B}(U)= U BU^\dagger\pmb\Pi U B U^\dagger+ U A U^T\pmb\Pi U^* AU^\dagger.
\end{align*}
Using Remark \ref{rem:wein} of Appendix \ref{sec:Wein}, we have $\mb{E}_{\Mr{U}}\mc{Q}(U)=0$. Thus,
\begin{align}
4\mb{E}_{\Mr{U}}\pmb\Gamma^2_m=\begin{pmatrix}
\pmb\Pi \mb{E}_{\Mr{U}}\mc{B}(U)\pmb\Pi& 0\\
0 & \pmb\Pi\mb{E}_{\Mr{U}}\mc{B}(U)\pmb\Pi
\end{pmatrix}.
\end{align}
We now compute the $\mb{E}_{\Mr{U}}\mc{B}(U)$. We have 
\begin{align*}
\mc{B}(U)= U BU^\dagger\pmb\Pi U B U^\dagger+ U A U^T\pmb\Pi U^* AU^\dagger.
\end{align*}

We first compute $\mb{E}_{\Mr{U}}U BU^\dagger\pmb\Pi U B U^\dagger$ as follows.
\begin{align*}
&\mb{E}_{\Mr{U}} U BU^\dagger\pmb\Pi U B U^\dagger\\
 &= \mb{E}_{\Mr{U}} \sum_{i_1,k_1,i_2,k_2,j_1,l_1,j_2,l_2} U_{i_1 k_1} B_{k_1 l_1} U^*_{j_1 l_1} \pmb\Pi _{j_1 i_2} U_{i_2 k_2} B_{k_2 l_2} U^*_{j_2l_2}\op{i_1}{j_2}\\
&= \sum_{i_1,k_1,i_2,k_2,j_1,l_1,j_2,l_2}  B_{k_1 l_1}  \pmb\Pi _{j_1 i_2}  B_{k_2 l_2}  \sum_{\alpha,\beta\in S_2}\prod_{x=1}^2\delta_{i_xj_{\alpha(x)}}\prod_{y=1}^2\delta_{k_y l_{\beta(y)}}\Mr{Wg}(2n,\alpha^{-1}\beta)\op{i_1}{j_2}\\
&= \sum_{i_1,k_1,i_2,k_2,j_1,l_1,j_2,l_2}  B_{k_1 l_1}  \pmb\Pi _{j_1 i_2}  B_{k_2 l_2}  \op{i_1}{j_2}\left[\delta_{i_1j_1}\delta_{i_2j_2}\delta_{k_1 l_1} \delta_{k_2 l_2}\Mr{Wg}(2n,(1)(2))+\delta_{i_1j_1}\delta_{i_2j_2}\delta_{k_1 l_2} \delta_{k_2 l_1}\Mr{Wg}(2n,(12)) \right.\\
&\left. + \delta_{i_1j_2}\delta_{i_2j_1}\delta_{k_1 l_1} \delta_{k_2 l_2}\Mr{Wg}(2n,(12)) + \delta_{i_1j_2}\delta_{i_2j_1}\delta_{k_1 l_2} \delta_{k_2 l_1}\Mr{Wg}(2n,(1)(2))\right]\\
&= \sum_{i_1,k_1,i_2,k_2}  \left[ B_{k_1 k_1}  \pmb\Pi _{i_1 i_2}  B_{k_2 k_2}  \op{i_1}{i_2}\Mr{Wg}(2n,(1)(2))+ B_{k_1 k_2}  \pmb\Pi _{i_1 i_2}  B_{k_2 k_1}  \op{i_1}{i_2}\Mr{Wg}(2n,(12)) \right.\\
&\left. + B_{k_1 k_1}  \pmb\Pi _{i_2 i_2}  B_{k_2 k_2}  \op{i_1}{i_1}\Mr{Wg}(2n,(12)) + B_{k_1 k_2}  \pmb\Pi _{i_2 i_2}  B_{k_2 k_1}  \op{i_1}{i_1} \Mr{Wg}(2n,(1)(2))\right]\\
&=   \left(\Mr{Tr}[B]\right)^2  \pmb\Pi \Mr{Wg}(2n,(1)(2))+ \Mr{Tr}[B^2]  \pmb\Pi \Mr{Wg}(2n,(12))+ \left(\Mr{Tr}[B]\right)^2 \Mr{Tr}[\pmb\Pi ] \mb{I}_{2n} \Mr{Wg}(2n,(12)) + \Mr{Tr}[B^2] \Mr{Tr}[\pmb\Pi ] \mb{I}_{2n}  \Mr{Wg}(2n,(1)(2))\\
&=  \left(\Mr{Tr}[B]\right)^2\left(  \pmb\Pi \Mr{Wg}(2n,(1)(2))+ m\mb{I}_{2n} \Mr{Wg}(2n,(12)) \right)+ \Mr{Tr}[B^2] \left(  \pmb\Pi \Mr{Wg}(2n,(12)) + m\mb{I}_{2n}  \Mr{Wg}(2n,(1)(2))  \right)\\
&=  \left(\Mr{Tr}[B]\right)^2  \frac{ 2n \pmb\Pi - m\mb{I}_{2n}  }{2n(4n^2-1)} + \Mr{Tr}[B^2] \frac{ 2 mn\mb{I}_{2n} - \pmb\Pi  }{2n(4n^2-1)}.
\end{align*}
Now we compute $\mb{E}_{\Mr{U}}U A U^T\pmb\Pi U^* AU^\dagger$ as follows.
\begin{align*}
&\mb{E}_{\Mr{U}} U A U^T\pmb\Pi U^* AU^\dagger\\
 &= \mb{E}_{\Mr{U}} \sum_{i_1,k_1,i_2,k_2,j_1,l_1,j_2,l_2} U_{i_1 k_1} A_{k_1 k_2} U_{i_2 k_2} \pmb\Pi _{i_2 j_1} U^*_{j_1l_1} A_{l_1 l_2} U^*_{j_2l_2}\op{i_1}{j_2}\\
&= \sum_{i_1,k_1,i_2,k_2,j_1,l_1,j_2,l_2}  A_{k_1 k_2}  \pmb\Pi _{i_2 j_1}  A_{l_1 l_2}  \sum_{\alpha,\beta\in S_2}\prod_{x=1}^2\delta_{i_xj_{\alpha(x)}}\prod_{y=1}^2\delta_{k_y l_{\beta(y)}}\Mr{Wg}(2n,\alpha^{-1}\beta)\op{i_1}{j_2}\\
&= \sum_{i_1,k_1,i_2,k_2,j_1,l_1,j_2,l_2}  A_{k_1 k_2}  \pmb\Pi _{i_2 j_1}  A_{l_1 l_2}  \op{i_1}{j_2}\left[\delta_{i_1j_1}\delta_{i_2j_2}\delta_{k_1 l_1} \delta_{k_2 l_2}\Mr{Wg}(2n,(1)(2))+\delta_{i_1j_1}\delta_{i_2j_2}\delta_{k_1 l_2} \delta_{k_2 l_1}\Mr{Wg}(2n,(12)) \right.\\
&\left. + \delta_{i_1j_2}\delta_{i_2j_1}\delta_{k_1 l_1} \delta_{k_2 l_2}\Mr{Wg}(2n,(12)) + \delta_{i_1j_2}\delta_{i_2j_1}\delta_{k_1 l_2} \delta_{k_2 l_1}\Mr{Wg}(2n,(1)(2))\right]\\
&= \sum_{i_1,k_1,i_2,k_2}  \left[  A_{k_1 k_2}  \pmb\Pi _{i_2 i_1}  A_{k_1 k_2}  \op{i_1}{i_2} \Mr{Wg}(2n,(1)(2))+ A_{k_1 k_2}  \pmb\Pi _{i_2 i_1}  A_{k_2 k_1}  \op{i_1}{i_2}\Mr{Wg}(2n,(12)) \right.\\
&\left. + A_{k_1 k_2}  \pmb\Pi _{i_2 i_2}  A_{k_1 k_2}  \op{i_1}{i_1}\Mr{Wg}(2n,(12)) + A_{k_1 k_2}  \pmb\Pi _{i_2 i_2}  A_{k_2 k_1}  \op{i_1}{i_1} \Mr{Wg}(2n,(1)(2))\right]\\
&=  \Mr{Tr}[A^2]\left(\pmb\Pi \Mr{Wg}(2n,(1)(2)) + \pmb\Pi \Mr{Wg}(2n,(12)) +\Mr{Wg}(2n,(12))  m\mb{I} +\Mr{Wg}(2n,(1)(2))  m\mb{I} \right)  \\
&=  \Mr{Tr}[A^2]\frac{1 }{2n(2n+1)} \left(\pmb\Pi + m\mb{I} \right).
\end{align*}
Thus,
\begin{align*}
\mb{E}_{\Mr{U}} \mc{B}(U)=  \frac{ (2n \pmb\Pi - m\mb{I}_{2n} ) \left(\Mr{Tr}[B]\right)^2  }{2n(4n^2-1)} +  \frac{ (2 mn\mb{I}_{2n} - \pmb\Pi )\Mr{Tr}[B^2] }{2n(4n^2-1)}+  \frac{\Mr{Tr}[A^2] }{2n(2n+1)} \left(\pmb\Pi + m\mb{I} \right).
\end{align*}
And we then get
%\begin{align}
%\label{eq:average-B}
%\pmb\Pi\mb{E}_{\Mr{U}} \mc{B}(U)\pmb\Pi=  \frac{ (2n - m ) \left(\Mr{Tr}[B]\right)^2 +  (2 mn -1 )\Mr{Tr}[B^2]  + (m+1)\Mr{Tr}[A^2] }{2n(4n^2-1)} \pmb\Pi.
%\end{align}
%
%
%
%
%
%We now compute $\mb{E}_{\Mr{U}}\mc{B}(U)$ (see appendix \ref{append-avg}) to get
\begin{align*}
&\pmb\Pi\mb{E}_{\Mr{U}}\mc{B}(U)\pmb\Pi=\left( \frac{(2n-m)}{2n(4n^2-1)} \left(\Mr{Tr}[B]\right)^2 + \frac{(2mn-1)}{2n(4n^2-1)}\Mr{Tr}\left[B^2\right]  +\frac{\left(m+1\right)\Mr{Tr}\left[A^2\right]}{2n(2n+1)} \right)\pmb\Pi.
\end{align*}
For large $n$, we have
\begin{align}
&\Mr{Tr}\left[\pmb\Pi\mb{E}_{\Mr{U}}\mc{B}(U)\pmb\Pi\right]=\frac{m}{4n^2}\left( 1+ \ord{-1}\right) \left(\Mr{Tr}[B]\right)^2 +\ord{-2}\Mr{Tr}\left[B^2\right]  +\ord{-2}\Mr{Tr}\left[A^2\right].
\end{align}
Using  $\Norm{B}{1} =\Mr{Tr}[B]=4n\nu_{\tm}$, $\Norm{B}{\infty}=O\left(n^{\beta}\right)$, and $
\Mr{Tr}\left[A^2\right]\leq \Mr{Tr}\left[B^2\right]\leq \Norm{B}{1}\Norm{B}{\infty}= \nu_{\tm} O\left(n^{\beta+1}\right)$, we have
\begin{align}
\Mr{Tr}\left[\pmb\Pi\mb{E}_{\Mr{U}}\mc{B}(U)\pmb\Pi\right]
&=4m \nu_{\tm}^2 \left( 1+ \ord{-1}\right)+\nu_{\tm} \ord{\beta-1} \nonumber\\
&=4m \nu_{\tm}^2  +  \ord{\beta-1}.
\end{align}
Thus, we have
\begin{align}
\label{eq:av2}
4\Mr{Tr}\left[\mb{E}_{\Mr{U}}\pmb\Gamma^2_m\right]
&=8m\nu_{\tm}^2+ O\left(n^{\beta-1}\right).
\end{align}
Combining Eqs. \eqref{eq:av1} and \eqref{eq:av2}, we have
\begin{align*}
\mb{E}_{\Mr{U}}T(U)&=\mb{E}_{\Mr{U}}\Mr{Tr}\left[\pmb\Gamma^2_m\right]-2\nu_{\tm}\mb{E}_{\Mr{U}}\Mr{Tr}\left[\pmb\Gamma_m\right]+2m\nu^2_{\tm}\\
&=2m\nu_{\tm}^2 + O\left(n^{\beta-1}\right) - 4m\nu_{\tm}^2+2m\nu_{\tm}^2\\
&=O\left(n^{\beta-1}\right).
\end{align*}
Thus, there exists a universal constant $\gamma>0$ such that $\mb{E}_{\Mr{U}}T(U)\leq \gamma n^{\beta-1}$. 

Next, we bound the Lipschitz constant for the function $T(U)$. Let $\pmb\Gamma_m(U)$ and $\pmb\Gamma_m(V)$ be two covariance matrices generated via unitaries $U$ and $V$, respectively. Also, let us denote $\pmb\Gamma_m(U)$ by $\pmb\Gamma_m$ and $\pmb\Gamma_m(V)$ by $\wt{\pmb\Gamma}_m$. Then we have
\begin{align}
\label{eq:gammam1}
|T(U)-T(V)|
&\leq \left|\Mr{Tr}\left[\pmb\Gamma^2_m-\wt{\pmb\Gamma}^2_m\right]\right|+2\nu_{th}\left|\Mr{Tr}\left[\pmb\Gamma_m-\wt{\pmb\Gamma}_m\right]\right|\nonumber\\
&\leq \Norm{\pmb\Gamma^2_m-\wt{\pmb\Gamma}^2_m}{1}+2\nu_{th}\Norm{\pmb\Gamma_m-\wt{\pmb\Gamma}_m}{1}\nonumber\\
&\leq \left(\Norm{\pmb\Gamma_m}{\infty}+ \Norm{\wt{\pmb\Gamma}_m}{\infty}+2\nu_{th}\right)\Norm{\pmb\Gamma_m-\wt{\pmb\Gamma}_m}{1}\nonumber\\
&\leq 2\Norm{\wt{J}_{2n}({\bf z})}{\infty}\Norm{\pmb\Gamma_m-\wt{\pmb\Gamma}_m}{1}\nonumber\\
&\leq 2\sqrt{2m}\Norm{\wt{J}_{2n}({\bf z})}{\infty}\Norm{\pmb\Gamma_m-\wt{\pmb\Gamma}_m}{2},
\end{align}
where we have used $\max\left\{\Norm{\pmb\Gamma_m}{\infty},\Norm{\wt{\pmb\Gamma}_m}{\infty}\right \} \leq \Norm{\wt{J}_{2n}({\bf z})}{\infty}/2$ and $2 \nu_{th} \leq \Norm{\wt{J}_{2n}({\bf z})}{\infty}$.
Further, we have
\begin{align}
\label{eq:gammam2}
\Norm{\pmb\Gamma_m-\wt{\pmb\Gamma}_m}{2}
&\leq \frac{1}{2}	\Norm{ F(U) \wt{J}_{2n}({\bf z}) F(U)^T -  F(V) \wt{J}_{2n}({\bf z}) F(V)^T }{2}\nonumber\\
&\leq	 \frac{1}{2} \Norm{ F(U) \wt{J}_{2n}({\bf z})\left(F(U)-F(V)\right)^T}{2}+\frac{1}{2}\Norm{ \left(F(U)-F(V)\right)\wt{J}_{2n}({\bf z}) F(V)^T }{2}\nonumber\\
&\leq	 \Norm{ \wt{J}_{2n}({\bf z})}{\infty} \Norm{F(U)-F(V)}{2}\nonumber\\
&=\Norm{ \wt{J}_{2n}({\bf z})}{\infty} \Norm{U\oplus U^*-V\oplus V^*}{2}\nonumber\\
&\leq	2\Norm{ \wt{J}_{2n}({\bf z})}{\infty} \Norm{U-V}{2}.
\end{align}
Thus,
\begin{align*}
|T(U)-T(V)|&\leq 4\sqrt{2m}\Norm{ \wt{J}_{2n}({\bf z})}{\infty}^2\Norm{U-V}{2}\\
&= \ord{2\beta}\Norm{U-V}{2},
\end{align*}
where we have used the fact that $\Norm{ \wt{J}_{2n}({\bf z})}{\infty}=\ord{\beta} $. Thus, the Lipschitz constant $L$ for the function $T(U)$ is equal to $\ord{2\beta}$. Now, we use concentration of measure phenomenon to the function $T(U) $ of random  unitaries $U$. Let us take
\begin{align*}
\epsilon > 2 \gamma n^{\beta-1},
\end{align*}
where $\gamma$ is a universal constant. Then we have
\begin{align*}
\Mr{Pr}\left[T(U)>\epsilon\right]&\leq \Mr{Pr}\left[T(U)>\frac{\epsilon}{2}+\mb{E}_{\Mr{U}}T(U)\right]\\
&\leq \exp\left[-\frac{n\epsilon^2}{48L^2}\right]\\
&\leq \exp\left[-\wt{\gamma} \epsilon^2 n^{1-4\beta}\right],
\end{align*}
where the second inequality follows from the concentration of measure phenomenon (see Appendix \ref{sec:lip}) and $\wt{\gamma} $ is a suitable universal constant. This concludes the proof of the Lemma \ref{lem:reg-spectra}.
\end{proof}

\section{Proof of Lemma \ref{lem:symp-sp}}
\label{sec:App-Lemma2}
Here we provide a proof of Lemma \ref{lem:symp-sp}. In a similar way as Lemma \ref{lem:reg-spectra}, Lemma \ref{lem:symp-sp}  establishes the typicality of the symplectic eigenspectra (see also Ref. \cite{Fukuda2019b}). Let us consider a function $\mf{T}_m:\Mr{U}(2n)\rightarrow\mb{R}$ of random unitary matrices defined as
\begin{align}
\mf{T}_m(U)&:=\Mr{Tr}\left[\left((\pmb\Omega\pmb\Gamma_m)^2 + \nu_{th}^2\mb{I}_{2m}\right)^2\right]=2\sum_{k=1}^{m}\left(-\nu_k^2 + \nu_{th}^2\right)^2,
\end{align}
where $\{\nu_k\}_{k=1}^m$ are the symplectic eigenvalues of $\pmb\Gamma_m$ and $\{\pm \mf{i} \nu_k\}_{k=1}^m$ comprises the spectra of matrix $\pmb\Omega\pmb\Gamma_m$. Also, $\pmb\Gamma_m\equiv \pmb\Gamma_m(U)$ is defined by Eq. \eqref{eq:new-cov6} and $\nu_\tm=\Tr[\wt{J}_{2n}({\bf z})]/(8n)$ as before. 

\medskip
\noindent
{\it Restatement of Lemma \ref{lem:symp-sp} [Ref. \cite{Fukuda2019b}]:--}
%\label{lem:symp-sp}
Let $\pmb\Gamma_{m}$ be the covariance matrix as in Eq.  \eqref{eq:new-cov6}. Further, assume that $8\beta<1$. For universal constants $C,c>0$ such that $\epsilon> Cn^{\beta-1}$, the symplectic eigenvalues $\{\nu_i\}_{i=1}^m$ of $\pmb\Gamma_m$ converge in probability to $\nu_{\tm}$, i.e.,
\begin{align}
\Mr{Pr}\left[ \sum_{i=1}^m\left(\nu_i^2-\nu_{\tm}^2\right)^2>\epsilon \right]\leq \exp\left[ -c\epsilon^2  n^{1-8\beta} \right].
\end{align}
%\end{lemma}

The proof of the above lemma follows from the concentration of measure phenomenon applied to $\mf{T}_m(U)$. The key steps include the calculation of the Lipschitz constant for $\mf{T}_m(U)$ and its average with respect to the unitaries. For completeness, we show that $\mb{E}_U\mf{T}_m(U)=\ord{\beta-1}$ and the Lipschitz constant for $\mf{T}_m(U)$ is given by $\ord{4\beta}$. Note that these results easily follow from Ref. \cite{Fukuda2019b}. 

%\medskip
%\noindent
%{\bf Proof of Lemma \ref{lem:symp-sp}:--}
\begin{proof}
We first compute the average of the function $\mf{T}_m(U)$ over random unitaries. Following Ref. \cite{Fukuda2019b}, we have
\begin{align*}
4\mb{E}_U\Mr{Tr}\left[(\pmb\Omega\pmb\Gamma_m)^2\right]
&=-2m\left[ \frac{2n-m}{2n(4n^2-1)} \left(\Mr{Tr}[B]\right)^2-\frac{m+1}{2n(2n+1)} \Mr{Tr}[A^2] + \frac{mn-1}{2n(4n^2-1)} \Mr{Tr}[B^2]\right]\\
&=-2m\left[ \frac{1}{4n^2}\left(1+\ord{-1}\right)  \left(\Mr{Tr}[B]\right)^2 + \ord{-2} \Mr{Tr}[A^2] + \ord{-2} \Mr{Tr}[B^2]\right].
\end{align*}
Using  $\Norm{B}{1} =\Mr{Tr}[B]=4n\nu_{\tm}$, $\Norm{B}{\infty}=O\left(n^{\beta}\right)$, and $
\Mr{Tr}\left[A^2\right]\leq \Mr{Tr}\left[B^2\right]\leq \Norm{B}{1}\Norm{B}{\infty}= \nu_{\tm} O\left(n^{\beta+1}\right)$, we have
\begin{align*}
&\mb{E}_U\Mr{Tr}\left[(\pmb\Omega\pmb\Gamma_m)^2\right]=-2m\nu_{\tm}^2  + \ord{\beta-1}.
\end{align*}
Similarly, following Ref. \cite{Fukuda2019b}, we have
\begin{align*}
&\mb{E}_U\Mr{Tr}\left[(\pmb\Omega\pmb\Gamma_m)^4\right]=2m\nu_{\tm}^4  + \ord{\beta-1}.
\end{align*}
Thus, 
\begin{align}
\mb{E}_U\mf{T}_m(U)=\ord{\beta-1}.
\end{align}

Now, we compute the Lipschitz constant for the function $\mf{T}_m(U)$. Let $\pmb\Gamma_m(U) \equiv\pmb\Gamma_m$ and $\pmb\Gamma_m(V)\equiv \wt{\pmb\Gamma}_m$. Then, again from Ref. \cite{Fukuda2019b}, we have 
\begin{align*}
|\mf{T}(U)-\mf{T}(V)|
&\leq \left(4\Norm{\wt{J}_{2n}({\bf z})}{\infty}^3 + 4\nu_{th}^2 \Norm{\wt{J}_{2n}({\bf z})}{\infty} \right) \Norm{\pmb\Gamma_m-\wt{\pmb\Gamma}_m}{1}\\
&\leq 5\Norm{\wt{J}_{2n}({\bf z})}{\infty}^3 \Norm{\pmb\Gamma_m-\wt{\pmb\Gamma}_m}{1}\\
&\leq 5\sqrt{2m}\Norm{\wt{J}_{2n}({\bf z})}{\infty}^3\Norm{\pmb\Gamma_m-\wt{\pmb\Gamma}_m}{2},
\end{align*}
where we have used $\max\left\{\Norm{\pmb\Gamma_m}{\infty},\Norm{\wt{\pmb\Gamma}_m}{\infty}\right \} \leq \Norm{\wt{J}_{2n}({\bf z})}{\infty}/2$ and $2 \nu_{th} \leq \Norm{\wt{J}_{2n}({\bf z})}{\infty}$. In Lemma \ref{lem:reg-spectra}, we have proved $\Norm{\pmb\Gamma_m-\wt{\pmb\Gamma}_m}{2} \leq 2 \Norm{\wt{J}_{2n}({\bf z})}{\infty} \Norm{U-V}{2}$, therefore
\begin{align*}
|\mf{T}(U)-\mf{T}(V)|
&\leq 10\sqrt{2m}\Norm{\wt{J}_{2n}({\bf z})}{\infty}^4\Norm{U-V}{2}\\
&= \ord{4\beta}\Norm{U-V}{2},
\end{align*}
where we have used the fact that $\Norm{ \wt{J}_{2n}({\bf z})}{\infty}=\ord{\beta} $. Thus, the Lipschitz constant $L$ for the function $\mf{T}(U)$ is equal to $\ord{4\beta}$.

Now, we use concentration of measure phenomenon to the function $\mf{T}(U) $ of random  unitaries $U$. Let $C$ be a universal constant such that $\mb{E}_U\mf{T}_m(U)\leq C n^{\beta-1}$  and let $\epsilon > C n^{\beta-1}$. Then we have
\begin{align*}
\Mr{Pr}\left[\mf{T}(U)>2\epsilon\right]&\leq \Mr{Pr}\left[T(U)>\epsilon +\mb{E}_{\Mr{U}}T(U)\right]\\
&\leq \exp\left[-\frac{n\epsilon^2}{12L^2}\right]\\
&\leq \exp\left[-c \epsilon^2 n^{1-8\beta}\right],
\end{align*}
where $c$ is a suitable universal constant. This concludes the proof of the Lemma \ref{lem:symp-sp}.

\end{proof}

\appendix

\section{Norms, Lipschitz continuity and concentration of measure phenomenon}
\label{sec:lip}
\noindent
{\bf Matrix norms:--}
Let us consider a vector space $V_n$ of complex $n\times n$ matrices. Let $X, Y\in V_n$, then a matrix norm on $V_n$ is a real-valued non-negative function $\Norm{\cdot}{}:V_n\rightarrow \mb{R}$ satisfying the following properties:
\begin{enumerate}
\item $\Norm{X}{}\geq 0$ while the equality holds if and only if $X=0$. \item $\Norm{\alpha X}{} = |\alpha| \Norm{X}{}$ for all $\alpha\in\mb{C}$. \item $\Norm{X+Y}{}\leq \Norm{X}{} + \Norm{Y}{}$.
\item $\Norm{XY}{}\leq \Norm{X}{}\Norm{Y}{}$.
\end{enumerate}
The last property is called the submultiplicativity \cite{Horn1985}. An important family of matrix norms, called Schatten $p$-norms with $p\geq 1$, is defined as
\begin{align}
\Norm{X}{p}:=\left(\sum_{i=1}^n s_i^p(X)\right)^{1/p},
\end{align}
where $\{s_i\}$ are the singular values of $X\in V_n$. These norms are unitarily invariant, i.e., for unitaries $U,V\in V_n$, $\Norm{UXV}{p}=\Norm{X}{p}$. Of particular importance to us are the cases with $p=1,2,\infty$, which correspond to trace, Hilbert-Schmidt, and operator norms,  respectively. In particular,
\begin{subequations}
\begin{align}
&\Norm{X}{1}:=\Mr{Tr}\left[\sqrt{X^\dagger X}\right];\\
&\Norm{X}{2}:=\sqrt{\Mr{Tr}\left[X^\dagger X\right]};\\
&\Norm{X}{\infty}:=\max_{\vec{x}\neq 0}\frac{\Norm{X\vec{x}}{}}{\Norm{\vec{x}}{}},
\end{align}
\end{subequations}
where $\vec{x}$ is an $n$ dimensional vector and $\Norm{\cdot}{}$ is usual Euclidean norm for vectors. We list some of the relations between these norms that we will be using. Let $X\in V_n$, then
\begin{align}
\Norm{X}{1}\leq \sqrt{n}\Norm{X}{2} \leq n\Norm{X}{\infty}.
\end{align}
Moreover, for $X, Y, Z\in V_n$, we have
\begin{align}
\Norm{XYZ}{p}\leq \Norm{X}{\infty}\Norm{Y}{p}\Norm{Z}{\infty}.
\end{align}

\noindent
{\bf Lipschitz continuity:--}
Let us consider two metric spaces $(X,d_X)$ and $(Y,d_Y)$, where $d_X$ (or $d_Y$) denotes the metric on $X$ (or $Y$). A function $F:X\rightarrow Y$ is said to be a Lipschitz continuous function if for any $x,x'\in X$
\begin{align}
d_Y(F(x)-F(x'))\leq L ~ d_X(x,x'),
\end{align}
where the positive constant $L$ is called the Lipschitz constant \cite{Searcoid2007}. Note that any other constant $L'\geq L$ is also a valid Lipschitz constant. For this work, we are interested in functions $F:\Mr{U}(n)\rightarrow \mb{R}$, where $\Mr{U}(n)$ is the set of $n\times n$ unitary matrices and $\mb{R}$ is the set of real numbers. Such a function $F$ is a Lipschitz continuous function with Lipschitz constant $L$ if for any $U, V\in \Mr{U}(n)$ we have
\begin{align}
\Mod{F(U)-F(V)}\leq L \, \Norm{U-V}{2}.
\end{align}

\noindent
{\bf Concentration of the measure phenomenon:--}
The concentration of the measure phenomenon refers to the collective phenomenon of certain smooth functions defined over measurable vector spaces taking values close to their average values almost surely \cite{Ledoux2005}. There are various versions of concentration inequalities depending on the input measurable space and there are various ways to prove them. A very general technique to prove such inequalities is via logarithmic Sobolev inequalities together with the Herbst argument (this is also called the ``entropy method", see e.g. \cite{Ledoux2005, Anderson2009, Raginsky2013}). Since we are interested in functions on the unitary group $\Mr{U}(n)$, a particularly suitable concentration inequality is given as follows \cite{Meckes2013} (see also \cite{Fukuda2019b}):
\begin{theorem}[\cite{Meckes2013}]
Let $\Mr{U}(n)$ be the group of $n\times n$ unitary matrices which is equipped with the Hilbert-Schmidt norm. Let $F:\Mr{U}(n)\rightarrow \mb{R}$ be a Lipschitz continuous function with Lipschitz constant $L$. Then for any $\epsilon >0$
\begin{align}
\Pr\left[F(U)>\mb{E}_UF +\epsilon\right]<\exp[-\frac{n\epsilon^2}{12L^2}],
\end{align}
where $\mb{E}_U$ denotes the average with respect to Haar measure on $\Mr{U}(n)$. 
\end{theorem}

%%%%%%%%%%%%%%%

\section{Average over unitaries and Weingarten calculus}
\label{sec:Wein}

Computing averages over the Haar measure on the unitary group is an essential part for establishing concentration inequalities for functions on the unitary group. In this section, we present briefly a method of Ref. \cite{Collins2006} to compute averages (see also Ref. \cite{Zhang2014}). Let $\Mr{U}(n)$ be the group of $n\times n$ unitary matrices equipped with the normalized Haar measure and $S_d$ be the symmetric group of $d$ objects. Let $U_{ij}=\bra{i}U\ket{j}$ be the matrix elements of $U\in\Mr{U}(n)$ in the computational basis. Then we have the following formula for the averages:
\begin{align}
\label{eq:wein}
&\mb{E}_U\left[\prod_{a=1}^d U_{i_aj_a} \prod_{b=1}^dU_{i'_bj'_b}^*\right]=\sum_{\pi,\sigma\in S_d} \prod_{a=1}^d \delta_{i_a i'_{\pi(a)}} \prod_{b=1}^d \delta_{j_b j'_{\sigma(b)}}\Mr{Wg}\left(n,\pi^{-1}\sigma\right),
\end{align}
where $\pi$ and $\sigma$ are permutations and the function  $\Mr{Wg}\left(n,\pi^{-1}\sigma\right)$ is called the Weingarten function, defined as
\begin{align}
\Mr{Wg}\left(n,\pi\right)=\frac{1}{d!^2}\sum_{\substack{{\lambda\vdash d}\\{l(\lambda)\leq n}}}\frac{\chi^{\lambda}(\pi)(\chi^{\lambda}(1))^2}{s_{\lambda,n}(1)}.
\end{align}
In the above expression $\lambda$ is a Young tableaux and the sum is over all the Young tableaux with $d$ boxes and rows $l(\lambda)\leq n$. For a given $\lambda$, $\chi^{\lambda}$ is the character corresponding to the irreducible representation labeled as $\lambda$ of $S_d$. $s_{\lambda,n}(1)$ is the dimension of the representation of $\Mr{U}(n)$ corresponding to a tableaux $\lambda$. In this work, we will need to compute averages for $d=2$ case. In this case,
\begin{align}
&\Mr{Wg}\left(n,(1)(2)\right)=\frac{1}{n^2-1};\\
&\Mr{Wg}\left(n,(12)\right)=-\frac{1}{n(n^2-1)}.
\end{align}
\begin{remark}
\label{rem:wein}
From Eq. \eqref{eq:wein} if the number of $U$ terms is different than that of $U^*$, then the expectation in Eq. \eqref{eq:wein} is zero.
\end{remark}

\end{widetext}

\end{document}